\documentclass[twocolumn,amsthm]{autartx}
\usepackage{amssymb,amsmath}
\usepackage{newtxtext,newtxmath}
\usepackage{amsfonts}
\usepackage{tikz}
  \usetikzlibrary{arrows,automata,shapes}
  \usetikzlibrary{decorations.pathmorphing}
  \usetikzlibrary{backgrounds}
  \usetikzlibrary{calc}
\usepackage[caption=false,font=footnotesize]{subfig}
\usepackage[textwidth=1.9cm,color=green!10,textsize=footnotesize]{todonotes}

\begin{document}

\begin{frontmatter}
\title{Complexity of Deciding Detectability in Discrete Event Systems}

\author{Tom{\' a}{\v s}~Masopust}\ead{masopust{@}math.cas.cz}
\address{Institute of Mathematics, Czech Academy of Sciences, {\v Z}i{\v z}kova 22, 616 62 Brno, Czechia}


\begin{abstract}
  Detectability of discrete event systems (DESs) is a question whether the current and subsequent states can be determined based on observations. Shu and Lin designed a polynomial-time algorithm to check strong (periodic) detectability and an exponential-time (polynomial-space) algorithm to check weak (periodic) detectability. Zhang showed that checking weak (periodic) detectability is PSpace-complete. This intractable complexity opens a question whether there are structurally simpler DESs for which the problem is tractable. In this paper, we show that it is not the case by considering DESs represented as deterministic finite automata without non-trivial cycles, which are structurally the simplest deadlock-free DESs. We show that even for such very simple DESs, checking weak (periodic) detectability remains intractable. On the contrary, we show that strong (periodic) detectability of DESs can be efficiently verified on a parallel computer.
\end{abstract} 

\end{frontmatter}

\section{Introduction}
  The detectability problem of discrete event systems (DESs) is a question whether the current and subsequent states of a DES can be determined based on observations. The problem was introduced and studied by Shu et al.~\cite{ShuLin2011,ShuLinYing2007}. Detectability generalizes other notions studied in the literature~\cite{CainesGW1988,Ramadge1986}, such as stability of Ozveren and Willsky~\cite{OzverenW1990}. Shu et al. further argue that many practical problems can be formulated as the detectability problem for DESs.
  
  Four variants of detectability have been defined: strong and weak detectability and strong and weak periodic detectability~\cite{ShuLinYing2007}. Shu et al.~\cite{ShuLinYing2007} investigated detectability for deterministic DESs. A deterministic DES is modeled as a deterministic finite automaton with a set of initial states rather than a single initial state. The motivation for more initial states comes from the observation that it is often not known which state the system is initially in. They designed exponential algorithms to decide detectability of DESs based on the computation of the observer.
  Later, to be able to handle more problems, they extended their study to nondeterministic DESs and improved the algorithms for strong (periodic) detectability of nondeterministic DESs to polynomial time~\cite{ShuLin2011}. Concerning the complexity of deciding weak detectability, Zhang~\cite{Zhang17} showed that the problem is PSpace-complete and that PSpace-hardness holds even for deterministic DESs with all events observable.
  Shu and Lin~\cite{ShuLin2013} further extended strong detectability to delayed DESs and developed a polynomial-time algorithm to check strong detectability for delayed DESs.
  Yin and Lafortune~\cite{YinLafortune17} recently extended weak and strong detectability to modular DESs and showed that checking both strong modular detectability and weak modular detectability is PSpace-hard.

  Zhang's intractable complexity of deciding weak (periodic) detectability opens the question whether there are structurally simpler DESs for which tractability can be achieved. To tackle this question, we consider structurally the simplest deadlock-free DESs modeled as deterministic finite automata without non-trivial cycles, that is, every cycle is in the form of a self-loop in a state of the DES. We show that even for these very simple DESs, checking weak (periodic) detectability remains PSpace-complete, and hence the problem is intractable for all practical cases.

  On the other hand, we show that deciding strong (periodic) detectability of DESs is NL-complete. Since NL is the class of problems that can be efficiently parallelized~\cite{AroraBarak2009}, we obtain that strong (periodic) detectability can be efficiently verified on a parallel computer.

\section{Preliminaries and Definitions}
  For a set $A$, $|A$| denotes the cardinality of $A$ and $2^{A}$ its power set. An {\em alphabet\/} $\Sigma$ is a finite nonempty set with elements called {\em events}. A {\em word\/} over $\Sigma$ is a sequence of events of $\Sigma$. Let $\Sigma^*$ denote the set of all finite words over $\Sigma$; the {\em empty word\/} is denoted by $\varepsilon$. For a word $u \in \Sigma^*$, $|u|$ denotes its length. As usual, the notation $\Sigma^+$ stands for $\Sigma^*\setminus\{\varepsilon\}$.

  A {\em nondeterministic finite automaton\/} (NFA) over an alphabet $\Sigma$ is a structure $A = (Q,\Sigma,\delta,I,F)$, where $Q$ is a finite nonempty set of states, $I\subseteq Q$ is a nonempty set of initial states, $F \subseteq Q$ is a set of marked states, and $\delta \colon Q\times\Sigma \to 2^Q$ is a transition function that can be extended to the domain $2^Q\times\Sigma^*$ by induction. The {\em language recognized by $A$\/} is the set $L(A) = \{w\in \Sigma^* \mid \delta(I,w)\cap F \neq\emptyset\}$. Equivalently, the transition function is a relation $\delta \subseteq Q\times \Sigma \times Q$, where $\delta(q,a)=\{s,t\}$ denotes two transitions $(q,a,s)$ and $(q,a,t)$.

  The NFA $A$ is {\em deterministic\/} (DFA) if it has a unique initial state, i.e., $|I|=1$, and no nondeterministic transitions, i.e., $|\delta(q,a)|\le 1$ for every $q\in Q$ and $a \in \Sigma$. We say that a DFA $A$ over $\Sigma$ is {\em total\/} if its transition function is total, that is, $|\delta(q,a)|=1$ for every $q\in Q$ and $a\in\Sigma$. For DFAs, we identify singletons with their elements and simply write $p$ instead of $\{p\}$. Specifically, we write $\delta(q,a)=p$ instead of $\delta(q,a)=\{p\}$.

  A {\em discrete event system\/} (DES) is modeled as an NFA $G$ with all states marked. Hence we simply write $G=(Q,\Sigma,\delta,I)$ without specifying the set of marked states. The alphabet $\Sigma$ is partitioned into two disjoint subsets $\Sigma_o$ and $\Sigma_{uo}=\Sigma\setminus\Sigma_o$, where $\Sigma_o$ is the set of {\em observable events\/} and $\Sigma_{uo}$ the set of {\em unobservable events}. 
  
  The detectability problems are based on the observation of events, described by the projection $P\colon \Sigma^* \to \Sigma_o^*$. The {\em projection} $P\colon \Sigma^* \to \Sigma_o^*$ is a morphism defined by $P(a) = \varepsilon$ for $a\in \Sigma\setminus \Sigma_o$, and $P(a)= a$ for $a\in \Sigma_o$. The action of $P$ on a word $w=\sigma_1\sigma_2\cdots\sigma_n$ with $\sigma_i \in \Sigma$ for $1\le i\le n$ is to erase all events from $w$ that do not belong to $\Sigma_o$; namely, $P(\sigma_1\sigma_2\cdots\sigma_n)=P(\sigma_1) P(\sigma_2) \cdots P(\sigma_n)$. The definition can readily be extended to infinite words and languages.
  
  As usual when detectability is studied~\cite{ShuLin2011}, we make the following two assumptions on the DES $G=(Q,\Sigma,\delta,I)$:
  \begin{enumerate}
    \item\label{deadlock-free} $G$ is {\em deadlock free}, that is, for every state of the system, at least one event can occur. Formally, for every $q\in Q$, there is $\sigma \in \Sigma$ such that $\delta(q,\sigma)\neq\emptyset$.
  
    \item No loop in $G$ consists solely of unobservable events: for every $q\in Q$ and every $w \in \Sigma_{uo}^+$, $q\notin \delta(q,w)$.
  \end{enumerate}

  The set of infinite sequences of events generated by the DES $G$ is denoted by $L^\omega (G)$. Given $Q' \subseteq Q$, the set of all possible states after observing a word $t \in \Sigma_o^*$ is denoted by 
    $R(Q',t) = \cup_{w \in \Sigma^*, P(w) = t} \delta(Q',w)$.
  For $w \in L^\omega (G)$, we denote the set of its prefixes by $Pr(w)$.

  A {\em decision problem\/} is a yes-no question, such as ``Is an NFA $A$ deterministic?'' A decision problem is {\em decidable\/} if there exists an algorithm solving the problem. Complexity theory classifies decidable problems to classes based on the time or space an algorithm needs to solve the problem. The complexity classes we consider in this paper are NL, P, NP, and PSpace denoting the classes of problems solvable by a nondeterministic logarithmic-space, deterministic polynomial-time, nondeterministic polynomial-time, and deterministic polynomial-space algorithm, respectively. The hierarchy of classes is NL $\subseteq$ P $\subseteq$ NP $\subseteq$ PSpace. Which of the inclusions are strict is an open problem. The widely accepted conjecture is that all are strict. A decision problem is NL-complete (resp. NP-complete, PSpace-complete) if it belongs to NL (resp. NP, PSpace) and every problem from NL (resp. NP, PSpace) can be reduced to it by a deterministic logarithmic-space (resp. polynomial-time) algorithm.

\section{The Detectability Problems}
  In this section, we recall the definitions of the detectability problems~\cite{ShuLin2011}. Let $\Sigma$ be an alphabet, $\Sigma_o\subseteq\Sigma$ the set of observable events, and $P$ the projection from $\Sigma$ to $\Sigma_o$.
  
  \begin{defn}[Strong detectability]
    A DES $G=(Q,\Sigma,\delta,I)$ is {\em strongly detectable\/} with respect to $\Sigma_{uo}$ if we can determine, after a finite number of observations, the current and subsequent states of the system for all trajectories of the system, i.e., $(\exists n \in \mathbb{N})(\forall s \in L^\omega(G))(\forall t \in Pr(s)) |P(t)| > n \Rightarrow |R(I,P(t))| = 1$.
  \end{defn}

  \begin{defn}[Strong periodic detectability]
    A DES $G=(Q,\Sigma,\delta,I)$ is {\em strongly periodically detectable\/} with respect to $\Sigma_{uo}$ if we can periodically determine the current state of the system for all trajectories of the system, i.e., $(\exists n \in \mathbb{N})(\forall s \in L^\omega(G))(\forall t \in Pr(s))(\exists t' \in \Sigma^*) tt'\in Pr(s) \land |P(t')| < n \land |R(I,P(tt'))| = 1$.
  \end{defn}

  \begin{defn}[Weak detectability]
    A DES $G=(Q,\Sigma,\delta,I)$ is {\em weakly detectable\/} with respect to $\Sigma_{uo}$ if we can determine, after a finite number of observations, the current and subsequent states of the system for some trajectories of the system, i.e., $(\exists n \in \mathbb{N})(\exists s \in L^\omega(G))(\forall t \in Pr(s))|P(t)| > n \Rightarrow |R(I,P(t))| = 1$.
  \end{defn}

  \begin{defn}[Weak periodic detectability]
    A DES $G=(Q,\Sigma,\delta,I)$ is {\em weakly periodically detectable\/} with respect to $\Sigma_{uo}$ if we can periodically determine the current state of the system for some trajectories of the system, i.e., $(\exists n \in \mathbb{N})(\exists s \in L^\omega(G))(\forall t \in Pr(s))(\exists t' \in \Sigma^*) tt'\in Pr(s) \land |P(t')| < n \land |R(I,P(tt'))| = 1$.
  \end{defn}

  In what follows, we make often implicit use of the following lemma whose proof is obvious by definition.
  \begin{lem}\label{lem5}
    Let $G=(Q,\Sigma,\delta,I)$ be a DES and $P$ be the projection from $\Sigma$ to $\Sigma_o$. Let $P(G)=(Q,\Sigma_o,\delta',I)$ denote the DES obtained from $G$ by replacing every transition $(p,a,q)$ by $(p,P(a),q)$. Then $G$ is weak/strong (periodic) detectable with respect to $\Sigma_{uo}$ if and only if $P(G)$ is weak/strong (periodic) detectable with respect to $\emptyset$. \qed
  \end{lem}

\section{Complexity of Deciding Weak Detectability}
  To decide weak (periodic) detectability of a DES, Shu and Lin~\cite{ShuLin2011} construct the observer and prove that the DES is weakly detectable if and only if there is a reachable cycle in the observer consisting of singleton DES state sets, and that the DES is weakly periodically detectable if and only if there is a reachable cycle in the observer containing a singleton DES state set. Because of the construction of the observer, the algorithms are exponential. However, as pointed out by Zhang~\cite{Zhang17}, the algorithms require only polynomial space. 
  
  Zhang~\cite{Zhang17} further shows that deciding weak (periodic) detectability is PSpace-hard. His construction results in a deterministic DES with several initial states. Although the transitions are deterministic, the DES is not a DFA because of the non-unique initial state. We slightly improve Zhang's result.
  
  \begin{thm}\label{thm6}
    Deciding whether a deterministic DES over a binary alphabet is weakly (periodically) detectable is PSpace-complete.
  \end{thm}
  \begin{proof}
    Membership in PSpace is known~\cite{Zhang17}. To show hardness, we modify Zhang's construction reducing the finite automata intersection problem: given a sequence of total\footnote{In his construction, Zhang does not assume that $A_i$ are total, which makes his construction incorrect. The assumption that $A_i$ are total fixes this minor mistake.} DFAs $A_1,\ldots, A_n$ over a common alphabet $\Sigma$, the problem asks whether $\cap_{i=1}^{n} L(A_i) \neq \emptyset$. Without loss of generality, we may assume that $\Sigma=\{0,1\}$~\cite{Kozen77}.
    
    In every $A_i=(Q_i,\{0,1\},\delta_i,q_0^i,F_i)$, we replace every transition $(p,x,q)$ by two transitions $(p,0,p')$ and $(p',x,q)$. Intuitively, we encode $0$ as $00$ and $1$ as $01$. Let $A_i'=(Q_i\cup Q_i',\{0,1\},\delta_i',q_0^i,F_i)$ denote the resulting DFA, where $Q_i' = \{p' \mid p\in Q_i\}$. Notice that no transition under event $1$ is defined in states of $Q_i$ of $A_i'$. 

    Let $G=(\cup_{i=1}^{n} (Q_i \cup Q_i') \cup\{s_1,s_2\},\{0,1\},\delta,I)$ be a DES, where $s_1$ and $s_2$ are new states, $I=\{q_0^1,\ldots,q_0^n,s_2\}$, and $\delta$ is defined as follows. If $(p,x,q)\in\delta_i$, for an $i\in\{1,\ldots,n\}$, we add $(p,x,q)$ to $\delta$. For every $p\in \cup_{i=1}^{n} F_i$, we add the transition $(p,1,s_2)$ to $\delta$, and for every $p\in \cup_{i=1}^{n} (Q_i\setminus F_i)$, the transition $(p,1,s_1)$. We add transitions $(s_i,z,s_i)$ for $i\in\{1,2\}$ and $z\in\{0,1\}$; cf. Fig.~\ref{fig7} for an illustration. Then $G$ is deterministic and total.
    \begin{figure}
      \centering
      \begin{tikzpicture}[>=stealth',->,auto,shorten >=1pt,node distance=3cm,double distance=1pt,
        state/.style={circle,minimum size=0mm,inner sep=2pt,very thin,draw=black,initial text=}]

        \node [state,accepting]   (a1)  {$p$};
        \node [state]             (a3) [left of=a1,node distance=1cm] {$p'$};
        \node []                  (a2) [above of=a3,node distance=.7cm]  {$A_i$};
        \node [state]             (d1) [below of=a1,node distance=1cm] {$q$};
        \node [state]             (d2) [left of=d1,node distance=1cm] {$q'$};
        \node [state]             (2)  [below right of=d1,node distance=2cm] {$s_1$};
        \node [state,initial]     (3)  [above right of=a1,node distance=2cm] {$s_2$};

        \node [state,accepting]   (x1) [right of=a1,node distance=3cm] {};
        \node [state]             (x3) [right of=x1,node distance=1cm] {};
        \node []                  (x2) [above of=x3,node distance=.7cm]  {$A_j$};
        \node [state]             (y1) [below of=x1,node distance=1cm] {};
        \node [state]             (y2) [right of=y1,node distance=1cm] {};
        
        \path 
          (d1) edge node {$1$} (2)
          (2)  edge[loop above] node {$0,1$} (2)
          (a1) edge node[above,pos=.4] {$1$} (3)
          (a1) edge node[above] {$0$} (a3)
          (d1) edge node[above] {$0$} (d2)
          (3)  edge[loop above] node {$0,1$} (3)
          (y1) edge node[above] {$1$} (2)
          (x1) edge node[above,pos=.4] {$1$} (3)
          (x1) edge node[above] {$0$} (x3)
          (y1) edge node[above] {$0$} (y2)
        ;

        \begin{pgfonlayer}{background}
          \path (d2.south  -| a3.west)+(-0.2,-0.1)  node (a) {};
          \path (a2.north  -| d1.east)+(0.1,0.1)  node (b) {};
          \path (y2.south  -| y1.west)+(-0.2,-0.3)  node (x) {};
          \path (x2.north  -| x3.east)+(0.2,0.2)  node (y) {};
          \path[rounded corners, draw=black] (a) rectangle (b);
          \path[rounded corners, draw=black] (x) rectangle (y);
        \end{pgfonlayer}
      \end{tikzpicture}
      \caption{The illustration of the proof of Theorem~\ref{thm6}}
      \label{fig7}
    \end{figure}
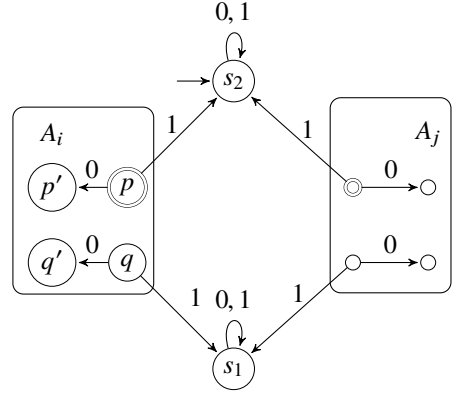

    We show that $\cap_{i=1}^{n} L(A_i)\neq\emptyset$ if and only if $G$ is weakly (periodically) detectable. If $w=a_1a_2\ldots a_m\in \cap_{i=1}^{n} L(A_i)$, then $\delta(I,0a_10a_2\ldots 0a_n 1 u) = \{s_2\}$ for every $u\in\{0,1\}^*$, and hence $G$ is weakly (periodically) detectable.
    On the other hand, if $\cap_{i=1}^{n} L(A_i)=\emptyset$, let $u\in\{0,1\}^*$. Let $u'$ be the longest prefix of $u$ of the form $(0(0+1))^*$, that is, $u'=0a_10a_2\cdots 0a_m$, for some $m\ge 0$. Then $\delta(I,u')=\cup_{i=1}^{n} \delta_i(q_0^i,a_1a_2\cdots a_m) \cup \{s_2\}=\{p_1,p_2,\ldots,p_n,s_2\}$, where $p_i\in Q_i$, and there is $i$ such that $p_i\notin F_i$. We now have three possibilities: (i) if $u=u'$, then the cardinality of $\delta(I,u)$ is $n+1$; (ii) if $u=u'0$, then $\delta(I,u)=\{p_1',p_2',\cdots,p_n',s_2\}$, where $p_i'\in Q_i'$. Again, the cardinality of $\delta(I,u)$ is $n+1$. Finally, (iii) if $u=u'1u''$, for some $u''\in\{0,1\}^*$, then $\delta(I,u'1u'')=\delta(\{p_1,\cdots,p_n,s_2\},1u'')=\{s_1,s_2\}$. In all cases, the cardinality of $\delta(I,u)$ is at least two, and hence $G$ is not weakly (periodically) detectable.
  \end{proof}
  
  \begin{cor}\label{cor7}
    Deciding whether a DES modeled as a DFA is weakly (periodically) detectable is PSpace-complete even if the DES has only three events, one of which is unobservable.
  \end{cor}
  \begin{proof}
    Consider the deterministic DES $G$ constructed in the proof of Theorem~\ref{thm6} with $n+1$ initial states denoted by $q_1,\ldots,q_{n+1}$. We construct $G'$ from $G$ by adding a new unobservable event $a$, the transitions $(q_i,a,q_{i+1})$, for $i=1,\ldots,n$, and by setting $q_1$ to be the sole initial state of $G'$. All other transitions of $G'$ coincide with those of $G$. Then $G'$ is a DFA and it is easy to see that the observers of $G$ with respect to $\emptyset$ and of $G'$ with respect to $\{a\}$ are identical, and hence $G$ is weakly (periodically) detectable if and only if $G'$ is.
  \end{proof}

  The unobservable event in the corollary is unavoidable because any DES modeled as a DFA with all events observable is always in a unique state, and hence trivially detectable. We now show that two observable events are also necessary for PSpace-hardness.
  \begin{thm}\label{thm7}
    Deciding whether a DES over a unary alphabet is weakly detectable is in P, and whether it is weakly periodically detectable is in NP.
  \end{thm}
  \begin{proof}
    Let $G=(Q,\{a\},\delta,I)$ be a DES with $n$ states. Then the observer of $G$ consists of a sequence of $k$ states followed by a cycle consisting of $\ell$ states, that is, the language of $G$ is $a^k (a^{\ell})^*$. Since the number of states of the observer of $G$ is at most $2^n$, $k+\ell \le 2^n$. 
    
    Now, $G$ is weakly detectable if and only if the states of the cycle of the observer of $G$ consist only of singleton states of $G$, that is, $|\delta(I,a^{k+i})|=1$ for $i=0,\ldots,\ell$ and $\ell \le n$, since the observer of $G$ has at most $n$ singleton sets. Since $\delta(I,a^{k+i})$, for $i=1,\ldots,n$, can be computed in polynomial time using the fast matrix multiplication, cf. Masopust~\cite{Masopust2018} for more details and an example, we can decide weak detectability in polynomial time. Namely, we compute $\delta(I,a^{2^n+i})$, for $i=1,\ldots,n$, check that $|\delta(I,a^{2^n+i})|=1$ and that there are $i<j$ such that $\delta(I,a^{2^n+i})=\delta(I,a^{2^n+j})$.
    
    Similarly, $G$ is weakly periodically detectable if and only if there is $m\le 2^n$ such that $|\delta(I,a^{2^n+m})|=1$. An NP algorithm can guess $m$ in binary and verify the guess in polynomial time by computing $\delta(I,a^{2^n+m})$ using the fast matrix multiplication.
  \end{proof}

  \subsection{Simplest DESs}
  Zhang's result gives rise to a question whether there are structurally simpler DESs with a tractable complexity of weak (periodic) detectability. The simplest DESs are acyclic DFAs, recognizing finite languages. Acyclic DESs are not deadlock-free. To fulfill deadlock-freeness, we consider DESs modeled as DFAs with cycles only in the form of self-loops. Such DFAs recognize a strict subclass of regular languages strictly included in {\em star-free languages}~\cite{BrzozowskiF80}. Star-free languages are languages definable by {\em linear temporal logic\/} widely used as a specification language in automated verification.

  Let $A=(Q,\Sigma,\delta,I,F)$ be an NFA. The reachability relation $\le$ on the state set $Q$ is defined by $p\le q$ if there is $w\in \Sigma^*$ such that $q\in \delta(p,w)$. The NFA $A$ is {\em partially ordered (poNFA)\/}  if the reachability relation $\le$ is a partial order. If $A$ is a DFA, we use the notation {\em poDFA}.
 
  A \emph{restricted partially ordered NFA (rpoNFA)} is a poNFA that is self-loop deterministic in the sense that the pattern of Fig.~\ref{fig_bad_pattern} does not appear. Formally, for every state $q$ and every event $a$, if $q\in \delta(q,a)$ then $\delta(q,a) = \{q\}$.
  \begin{figure}
    \centering
    \begin{tikzpicture}[baseline,->,>=stealth,auto,shorten >=1pt,node distance=2cm,
      state/.style={circle,minimum size=0mm,inner sep=2pt,very thin,draw=black,initial text=}]
      \node[state]  (a) {};
      \node[state]  (aa) [right of=a]  {};
      \path
        (a) edge[loop above] node {$a$} (a)
        (a) edge node {$a$} (aa)
        ;
    \end{tikzpicture}
    \caption{The forbidden pattern of rpoNFAs}
    \label{fig_bad_pattern}
  \end{figure}
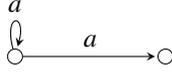
  This condition is trivially satisfied by poDFAs and it is known that poDFAs and rpoNFAs recognize the same class of languages~\cite{mfcs16:mktmmt_full}.

  We show that deciding weak (periodic) detectability remains PSpace-complete even if the DES is modeled as a poDFA, or as an rpoNFA with all events observable. Consequently, the problem is intractable for all practical cases.

  Recall that Zhang~\cite{Zhang17} obtained his result by reducing the finite automata intersection problem. Since his construction does not introduce any non-trivial cycles, it could seem that it also shows the result for poDFAs. This is, however, not the case because the complexity of the intersection problem for poDFAs is not known. Therefore, to prove our results, we need to use a different technique, namely the {\em universality problem\/} for rpoNFAs. The problem asks, given an rpoNFA $A$ over $\Sigma$, whether $L(A)=\Sigma^*$. It is PSpace-complete in general and coNP-complete if the alphabet is fixed a priori~\cite{mfcs16:mktmmt_full}.

  \begin{thm}\label{thm5}
    Deciding weak (periodic) detectability of DESs modeled as rpoNFAs is PSpace-complete even if all events are observable.
  \end{thm}
  \begin{proof}
    Membership in PSpace holds for general NFAs~\cite{Zhang17}. To show hardness, we reduce the non-universality problem for rpoNFAs. Let $A=(Q,\Sigma,\delta,I,F)$ be an rpoNFA. We construct an rpoNFA $A'=(Q\cup\{\clubsuit,\star\},\Sigma\cup\{\diamond\},\delta', I \cup \{\clubsuit\}, F)$, where $\clubsuit$ and $\star$ are new states and $\diamond$ is a new event. Let $\delta' = \delta$. We extended $\delta'$ as follows. For every non-marked state $p$ of $A$, we add the transition $(p,\diamond,\clubsuit)$ to $\delta'$, and for every marked state $p$ of $A$, we add the transition $(p,\diamond,\star)$ to $\delta'$. We add transitions $(\clubsuit,a,\clubsuit)$ and $(\star,a,\star)$ to $\delta'$ for every event $a$ of $A'$. The construction is illustrated in Fig.~\ref{fig5}. All events are observable. We show that $A$ is non-universal if and only if $A'$ is weakly (periodically) detectable.

    \begin{figure}
      \centering
      \begin{tikzpicture}[>=stealth',->,auto,shorten >=1pt,node distance=3cm,double distance=1pt,
        state/.style={circle,minimum size=0mm,inner sep=2pt,very thin,draw=black,initial text=}]

        \node []      (a1)  {Marked};
        \node []      (d1) [above of=a1, node distance=.2cm]  {};
        \node []      (e1) [below of=a1, node distance=.2cm]  {};
        \node []      (d2) [left  of=d1, node distance=.7cm]  {};
        \node []      (e2) [right of=e1, node distance=.7cm]  {};
        \node []      (a2) [below of=e1, node distance=.5cm]  {$A$ over $\Sigma$};
        \node []      (d1) [right of=a1, node distance=2.4cm] {Non-marked};
        \node []      (b1) [above of=d1, node distance=.2cm]  {};
        \node []      (c1) [below of=d1, node distance=.2cm]  {};
        \node []      (b2) [left  of=b1, node distance=1cm]   {};
        \node []      (c2) [right of=c1, node distance=1cm]   {};
        \node [state,initial below] (2)  [right of=d1,node distance=2.5cm] {$\clubsuit$};
        \node [state] (3)  [left of=a1,node distance=2.2cm] {$\star$};

        \path 
          (d1) edge node[pos=.6] {$\diamond$} (2)
          (2) edge[loop above] node {$\Sigma\cup\{\diamond\}$} (2)
          (a1) edge node[above,pos=.6] {$\diamond$} (3)
          (3) edge[loop above] node {$\Sigma\cup\{\diamond\}$} (3)
        ;

        \begin{pgfonlayer}{background}
          \filldraw [line width=4mm,join=round,black!10]
            (c2.south  -| c2.east)  rectangle (b2.north  -| b2.west);
          \filldraw [line width=4mm,join=round,black!10]
            (e2.south  -| e2.east)  rectangle (d2.north  -| d2.west);
          \path (a2.south  -| d2.west)+(-0.3,0)   node (a) {};
          \path (b1.north  -| c2.east)+(0.3,0.3)  node (b) {};
          \path[rounded corners, draw=black] (a) rectangle (b);
        \end{pgfonlayer}
      \end{tikzpicture}
      \caption{The illustration of the proof of Theorem~\ref{thm5}}
      \label{fig5}
    \end{figure}
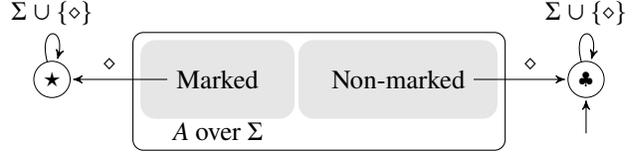

    If $A$ is not universal, there is a word $w\in\Sigma^* \setminus L(A)$ and $\delta(I,w)$ consists of non-marked states of $A$, i.e., $\delta(I,w)\cap F=\emptyset$. Then $\delta'(I\cup\{\clubsuit\},w\diamond)=\delta'(\delta(I,w)\cup\{\clubsuit\},\diamond)=\{\clubsuit\}$. Since $\delta(\clubsuit,u)=\clubsuit$ for every word $u$, we have that $A'$ is weakly (periodically) detectable.
        
    If $A$ is universal, we show that for every $w\in (\Sigma\cup\{\diamond\})^*$ the set $\delta'(I\cup\{\clubsuit\},w)$ has at least two elements, and therefore $A'$ is not weakly (periodically) detectable. If $w$ does not contain $\diamond$, then $\delta(I,w) \cap F \neq \emptyset$. Then $\delta'(I \cup \{\clubsuit\}, w) = \delta(I,w) \cup \{\clubsuit\}$ and since $\clubsuit\notin \delta(I,w)$, $|\delta(I,w) \cup \{\clubsuit\}|\ge 2$. If $w=w_1\diamond w_2$ with $w_1\in\Sigma^*$, then $\delta'(I \cup \{\clubsuit\}, w_1\diamond) = \delta'(\delta(I,w_1) \cup \{\clubsuit\}, \diamond) = \{\clubsuit,\star\}$ because $\delta(I,w_1) \cap F\neq\emptyset$ by the universality of $A$.
  \end{proof}

  We now show that intractability holds even if the DESs are modeled as poDFAs over a very small alphabet.

  \begin{thm}\label{thm9}
    Deciding weak (periodic) detectability of DESs modeled as poDFAs over the alphabet $\{0,1,\diamond,a,b\}$ with $a$ and $b$ unobservable is PSpace-complete.
  \end{thm}
  \begin{proof}
    Membership in PSpace holds for general NFAs. We show hardness by reducing the {\em non-universality problem for poNFAs}, which is PSpace-complete even if the alphabet is binary~\cite{mfcs16:mktmmt_full}. Let $A=(Q,\{0,1\},\delta,I,F)$ be a poNFA. We construct poDFA $D=(Q\cup\{\clubsuit,\star\}\cup Q',\{0,1,\diamond,a,b\},\delta',s,F)$, where $\clubsuit$ and $\star$ are new states, in the following steps. Initially, we define $\delta'=\delta$ and extend it as follows.
    
    First, for every non-marked state $p$ of $A$, we add the transition $(p,\diamond,\clubsuit)$ to $\delta'$, and for every marked state $p$ of $A$, we add the transition $(p,\diamond,\star)$. For every event $c\in\{0,1,\diamond,a,b\}$, we add the transitions $(\clubsuit,c,\clubsuit)$ and $(\star,c,\star)$; the result is similar to that illustrated in Fig.~\ref{fig5}. The result is a poNFA.
    
    Second, we convert the poNFA to an rpoNFA so that, for every state $p$ with two transitions $(p,x,p)$ and $(p,x,q)$, $p\neq q$, we replace the transition $(p,x,q)$ with two transitions $(p,x',p')$ and $(p',x,q)$, where $x'$ is a new event and $p'$ a new state; cf. Fig.~\ref{fig3} for an illustration. 
    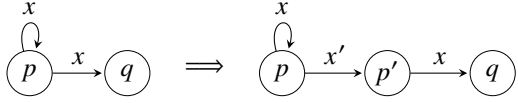
\begin{figure}
      \centering
      \begin{tikzpicture}[baseline,->,>=stealth,auto,shorten >=1pt,node distance=1.3cm,
        state/.style={circle,minimum size=6mm,inner sep=1,very thin,draw=black,initial text=}]
        \node[state]  (1) {$p$};
        \node[state]  (2) [right of=1] {$q$};
        \path
          (1) edge[loop above] node {$x$} (1)
          (1) edge node[pos=.5,sloped] {$x$} (2)
          ;
      \end{tikzpicture}
      \quad $\Longrightarrow$ \quad
      \begin{tikzpicture}[baseline,->,>=stealth,auto,shorten >=1pt,node distance=1.4cm,
        state/.style={circle,minimum size=6mm,inner sep=1,very thin,draw=black,initial text=}]
        \node[state]  (1) {$p$};
        \node[state]  (2) [right of=1] {$p'$};
        \node[state]  (3) [right of=2] {$q$};
        \path
          (1) edge[loop above] node {$x$} (1)
          (1) edge node[pos=0.5,sloped] {$x'$} (2)
          (2) edge node[pos=0.5,sloped] {$x$} (3)
          ;
      \end{tikzpicture}
      \caption{Conversion of poNFA to rpoNFA in Theorem~\ref{thm9}}
      \label{fig3}
    \end{figure}
    State $p'$ is added to $Q'$. We repeat this procedure until all such nondeterministic transitions are eliminated. The result is an rpoNFA.

    Third, we convert the rpoNFA to a poDFA as follows. For every state $p$ with two different transitions $(p,x,q)$ and $(p,x,r)$, we replace the transition $(p,x,q)$ by two transitions $(p,x',p')$ and $(p',x,q)$, where $x'$ is a new event and $p'$ a new state not in $D$; cf. Fig.~\ref{fig2} for an illustration. 
    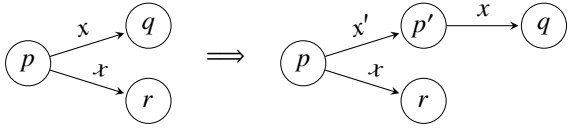
\begin{figure}
      \centering
      \begin{tikzpicture}[baseline,->,>=stealth,auto,shorten >=1pt,node distance=1.6cm,
        state/.style={circle,minimum size=6mm,inner sep=1,very thin,draw=black,initial text=}]
        \node[state]  (1) {$p$};
        \node         (0) [right of=1] {};
        \node[state]  (2) [above of=0,node distance=.5cm] {$q$};
        \node[state]  (3) [below of=0,node distance=.5cm] {$r$};
        \path
          (1) edge node[pos=.7,sloped] {$x$} (2)
          (1) edge node[pos=.4,sloped] {$x$} (3)
          ;
      \end{tikzpicture}
      \quad $\Longrightarrow$ \quad
      \begin{tikzpicture}[baseline,->,>=stealth,auto,shorten >=1pt,node distance=1.6cm,
        state/.style={circle,minimum size=6mm,inner sep=1,very thin,draw=black,initial text=}]
        \node[state]  (1) {$p$};
        \node         (0) [right of=1] {};
        \node[state]  (2) [above of=0,node distance=.5cm] {$p'$};
        \node[state]  (4) [right of=2] {$q$};
        \node[state]  (3) [below of=0,node distance=.5cm] {$r$};
        \path
          (1) edge node[pos=0.8,sloped] {$x'$} (2)
          (2) edge node[pos=0.5,sloped] {$x$} (4)
          (1) edge node[pos=0.4,sloped] {$x$} (3)
          ;
      \end{tikzpicture}
      \caption{The 'determinization' of $A'$; $x'$ and $p'$ are a new event and a new state, and hence different from those in Fig.~\ref{fig3}}
      \label{fig2}
    \end{figure}
    State $p'$ is added to $Q'$. We repeat this procedure until there is no state with two nondeterministic transitions. Because we started from an rpoNFA, $q \neq p \neq r$, and hence the newly added events do not occur in a self-loop. The initial states of this automaton are $I\cup\{\clubsuit\}$. The automaton is a deterministic DES.
    
    Let $\Gamma$ denote the set of all newly introduced events. We encode every event of $\Gamma$ as a binary word over $\{a,b\}$. To encode $|\Gamma|$ different events as binary words, it is sufficient to consider words of length $m=\lceil \log(|\Gamma|) \rceil$. Let $\text{enc}\colon \Gamma \to \{a,b\}^{m}$ be an arbitrary encoding (injective function). We replace every transition $(p,x',p')$ with $x'\in\Gamma$ by the sequence of transitions $(p,\text{enc}(x'),p')$ added to $\delta'$, which requires to add at most $m-2$ new states to $Q'$. 
    For instance, if $(p,x,p')$ and $(p,y,p'')$ are two transitions with $x,y\in\Gamma$, and $\text{enc}(x)=aab$ and $\text{enc}(y)=aba$, then the transition $(p,x,p')$ is replaced by the sequence of transitions $(p,a,p_1)$, $(p_1,a,p_2)$, $(p_2,b,p')$ added to $\delta'$, where $p_1$, $p_2$ are new states added to $Q'$, and the transition $(p,y,p'')$ is replaced by the sequence of transitions $(p,a,p_1)$, $(p_1,b,p_3)$, $(p_3,a,p'')$ added to $\delta'$ where $p_3$ is a new state added to $Q'$; cf. Fig.~\ref{fig6}. 
    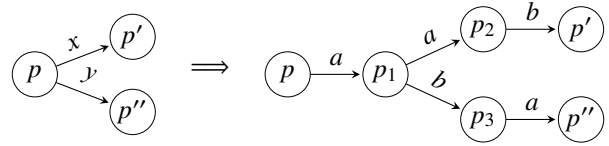
\begin{figure}
      \centering
      \begin{tikzpicture}[baseline,->,>=stealth,auto,shorten >=1pt,node distance=1.3cm,
        state/.style={circle,minimum size=6mm,inner sep=1,very thin,draw=black,initial text=}]
        \node[state]  (1) {$p$};
        \node         (0) [right of=1] {};
        \node[state]  (2) [above of=0,node distance=.5cm] {$p'$};
        \node[state]  (3) [below of=0,node distance=.5cm] {$p''$};
        \path
          (1) edge node[pos=.7,sloped] {$x$} (2)
          (1) edge node[pos=.2,sloped] {$y$} (3)
          ;
      \end{tikzpicture}
      \quad $\Longrightarrow$ \quad
      \begin{tikzpicture}[baseline,->,>=stealth,auto,shorten >=1pt,node distance=1.3cm,
        state/.style={circle,minimum size=6mm,inner sep=1,very thin,draw=black,initial text=}]
        \node[state]  (1) {$p$};
        \node[state]  (2) [right of=1] {$p_1$};
        \node (0) [right of=2] {};
        \node[state]  (3) [above of=0,node distance=.6cm] {$p_2$};
        \node[state]  (4) [below of=0,node distance=.6cm] {$p_3$};
        \node[state]  (5) [right of=3] {$p'$};
        \node[state]  (6) [right of=4] {$p''$};
        \path
          (1) edge node {$a$} (2)
          (2) edge node[pos=.8,sloped] {$a$} (3)
          (2) edge node[pos=.2,sloped] {$b$} (4)
          (3) edge node {$b$} (5)
          (4) edge node {$a$} (6)
          ;
      \end{tikzpicture}
      \caption{The encoding $\text{enc}(x)=aab$ and $\text{enc}(y)=aba$}
      \label{fig6}
    \end{figure}

    To obtain a single initial state, we proceed as follows. Let $q_1,\ldots,q_{n}$ denote the states of $I\cup\{\clubsuit\}$. Let $m=\lceil\log(n)\rceil$. We construct a binary tree of depth $m$ over $\{a,b\}$ and add it as depicted in Fig.~\ref{fig4}. The number of leaves of the tree is $2^m$ and the number of nodes of the tree is $2^{m+1}-1 = O(n)$. Let the leaves be denoted by $1,2,\ldots,2^m$. We add the transitions $(i,a,q_i)$ for $1\le i \le n$. The states of the tree are denoted by $Z$. The resulting automaton is a poDFA with polynomially many new events and states, and a single initial state $s$.
    \begin{figure}
      \centering
      \begin{tikzpicture}[baseline,auto,->,>=stealth,shorten >=1pt,node distance=1.5cm,
        state/.style={ellipse,minimum size=4mm,inner sep=0pt,very thin,draw=black,initial text=},
        every node/.style={font=\small}]
        \node[state]  (1) {$q_1$};
        \node[state]  (2) [below of=1,node distance=.5cm]  {$q_2$};
        \node[state]  (5) [below of=2,node distance=.5cm]  {$q_3$};
        \node[state]  (6) [below of=5,node distance=.5cm]  {$q_4$};
        \node[]       (3) [right of=1,node distance=1cm]  {};
        \node[state]  (10) [left of=1] {$1$};
        \node[state]  (11) [left of=2] {$2$};
        \node[state]  (12) [left of=5] {$3$};
        \node[state]  (13) [left of=6] {$4$};
        \node         (y) at ($(10)!0.5!(11)$) {};
        \node         (z) at ($(12)!0.5!(13)$) {};
        \node[state]  (20) [left of=y] {};
        \node[state]  (21) [left of=z] {};
        \node         (x) at ($(20)!0.5!(21)$) {};
        \node[state,initial] (30) [left of=x] {$s$};
        \node[]       (z) [above of=30,node distance=.8cm] {$Z$};
        \path
          (30) edge[pos=.7,sloped] node {{\tiny $a$}} (20)
          (30) edge[pos=.3,sloped] node {{\tiny $b$}} (21)
          (20) edge[pos=.8,sloped] node {{\tiny $a$}} (10)
          (20) edge[pos=.5,sloped] node {{\tiny $b$}} (11)
          (21) edge[pos=.8,sloped] node {{\tiny $a$}} (12)
          (21) edge[pos=.5,sloped] node {{\tiny $b$}} (13)
          (10) edge[pos=.5,sloped] node {{\tiny $a$}} (1)
          (11) edge[pos=.5,sloped] node {{\tiny $a$}} (2)
          (12) edge[pos=.5,sloped] node {{\tiny $a$}} (5)
          (13) edge[pos=.5,sloped] node {{\tiny $a$}} (6)
          ;
        \begin{pgfonlayer}{background}
          \path (1.north -| 3.east) + (0.1,0.1)     node (a) {};
          \path (6.south -| 1.west) + (-0.2,-0.1)   node (b) {};
          \path[rounded corners, draw=black] (a)    rectangle (b);
          \path (10.north -| 10.east) + (0.1,0.1)   node (a) {};
          \path (13.south -| 30.west) + (-0.2,-0.1) node (b) {};
          \path[rounded corners, draw=black] (a)    rectangle (b);
        \end{pgfonlayer}
      \end{tikzpicture}
      \caption{Construction of a single initial state illustrated for $4$ initial states}
      \label{fig4}
    \end{figure}
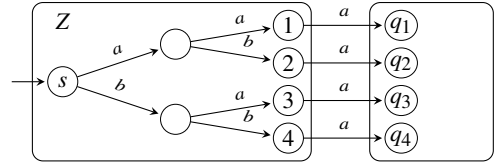

    Let $D$ be the resulting automaton. Then $D$ is a poDFA over the alphabet $\{0,1,\diamond,a,b\}$ of polynomial size with respect to the size of $A$ (i.e., the number of states, events, and transitions of $A$).
    Let $P$ be the projection from $\{0,1,\diamond,a,b\}$ to $\{0,1,\diamond\}$.
    Let $P(D)$ denote the poNFA obtained from $D$ by replacing every transition $(p,a,q)$ by $(p,P(a),q)$, and let $\delta''$ denote its transition relation. By Lemma~\ref{lem5}, $D$ is weakly (periodically) detectable with respect to $\{a,b\}$ if and only if $P(D)$ is weakly (periodically) detectable with respect to $\emptyset$. We show that $P(D)$ is weakly (periodically) detectable if and only if $A$ is not universal.

    If $A$ is not universal, there is $w\in\{0,1\}^*\setminus L(A)$ and $\delta(I,w)$ contains no marked state of $A$. Then $\delta''(I\cup\{\clubsuit\},w) = \delta(I,w) \cup Y \cup \{\clubsuit\}$ where $Y\subseteq Q'$ since every unobservable transition reachable from $I$ ends in a state from $Q'$. Then $\delta''(s,w\diamond) = \delta''(I\cup\{\clubsuit\},w\diamond) = \delta''(\delta(I,w)\cup Y\cup \{\clubsuit\},\diamond)=\{\clubsuit\}$ because $\delta''(Y,\diamond)=\emptyset$. Since $\delta(\clubsuit,u)=\clubsuit$ for every word $u$, $P(D)$ is weakly (periodically) detectable.
    
    If $A$ is universal, we show that for every $w\in \{0,1,\diamond\}^*$ the set $\delta''(s,w)$ has at least two elements, which shows that $P(D)$ is not weakly (periodically) detectable. Thus, if $\diamond$ does not occur in $w$, then $\delta(I,w)\cap F\neq\emptyset$. Since $\delta(I,w)\cup\{\clubsuit\}\subseteq \delta''(I\cup\{\clubsuit\},w) \subseteq \delta''(s,w)$ and $\clubsuit\notin \delta(I,w)$, $|\delta(I,w)\cup\{\clubsuit\}|\ge 2$. If $w=w_1\diamond w_2$ with $w_1\in\{0,1\}^*$, then $\delta(I,w_1) \cap F \neq \emptyset$ by the universality of $A$. Therefore, $\{\clubsuit,\star\}=\delta''(I \cup \{\clubsuit\},w_1\diamond) \subseteq \delta''(s,w)$.
  \end{proof}

\section{Complexity of Deciding Strong Detectability}
  Shu and Lin~\cite{ShuLin2011} designed a polynomial-time algorithm to decide strong (periodic) detectability, and hence the problem is in P. Is the problem P-complete or does it belong to NL? This question asks whether the problem can be efficiently parallelized (is in NL) or not (is P-complete)~\cite{AroraBarak2009}. We show that the problem is NL-complete and can thus be efficiently solved on a parallel computer.
  
  \begin{thm}\label{thm_sdnl-c}
    Deciding whether a DES is strongly (periodically) detectable is NL-complete.
  \end{thm}
  \begin{proof}
    For a DES $G$, Shu and Lin~\cite{ShuLin2011} construct an NFA $G_{det}$ of polynomial size whose states are subsets of states of $G$ of cardinality one or two (except for the initial state), such that $G$ is strongly detectable if and only if 
    (a) any state reachable from any loop in $G_{det}$ is of cardinality one;
    and strongly periodically detectable if and only if 
    (b) all loops in $G_{det}$ include at least one state of cardinality one. 

    We prove the membership in NL by showing that checking that the conditions do not hold is in NL. Since NL is closed under complement~\cite{Immerman88,Szelepcsenyi87}, checking that the conditions are satisfied is also in NL. 
    
    To check that (a) is not satisfied, the NL algorithm guesses two states of $G_{det}$, $x$ and $y$, where $y$ is of cardinality two, and verifies that
      $y$ is reachable from $x$, 
      $x$ is reachable from the initial state of $G_{det}$, and 
      $x$ is in a cycle, i.e., $x$ is reachable from $x$ by a path having at least one transition.
    
    To check that (b) is not satisfied, the NL algorithm guesses a state $x$ of $G_{det}$ of cardinality two and verifies that 
      $x$ is reachable from the initial state of $G_{det}$ and 
      that $x$ is in a cycle consisting only of states of cardinality two.
    
    For more details how to check reachability in NL, the reader is referred to the literature~\cite{Masopust2018}.
    
    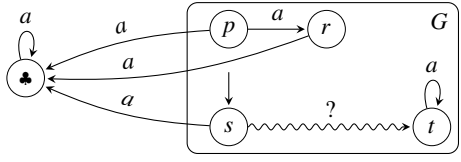
\begin{figure}
      \centering
      \begin{tikzpicture}[baseline,auto,->,>=stealth,shorten >=1pt,node distance=1.3cm,
        state/.style={ellipse,minimum size=5mm,inner sep=1pt,very thin,draw=black,initial text=},
        every node/.style={font=\small}]
        \node[state,initial above]  (1) {$s$};
        \node[state]          (2) [above of=1]  {$p$};
        \node                 (c) at ($(1)!0.5!(2)$) {};
        \node[state]          (5) [right of=2]  {$r$};
        \node[state]          (3) [right of=1,node distance=2.7cm]  {$t$};
        \node[state]          (4) [left of=c,node distance=2.7cm]  {$\clubsuit$};
        \node at (2.8,1.4) {$G$};
        \path
          (2) edge node[sloped,above] {$a$} (5)
          (1) edge[bend left=10] node[pos=0.5,sloped,above] {$a$} (4)
          (2) edge[bend right=10] node[pos=0.5,sloped,above] {$a$} (4)
          (4) edge[loop above] node[pos=0.5,sloped,above] {$a$} (4)
          (3) edge[loop above] node[pos=0.5,sloped,above] {$a$} (3)
          (5) edge[bend left=11] node[pos=0.7,sloped,above] {$a$} (4)
          (1) edge[style={decorate, decoration={snake,amplitude=.4mm,segment length=1.7mm,post length=1.3mm}}] node{?} (3) ;
          ;
        \begin{pgfonlayer}{background}
          \path (2.north -| 3.east) + (0.1,0.1)    node (a) {};
          \path (1.south -| 1.west) + (-0.3,-0.1)  node (b) {};
          \path[rounded corners, draw=black] (a) rectangle (b);
        \end{pgfonlayer}
      \end{tikzpicture}
      \caption{The DES $A$ from the NL-hardness proof of Theorem~\ref{thm_sdnl-c}}
      \label{fig1}
    \end{figure}

    To show NL-hardness, we reduce the {\em DAG non-reachability problem}~\cite{ChoH91}. Given a directed acyclic graph $G=(V,E)$ and nodes $s,t\in V$, it asks whether $t$ is not reachable from $s$. 
    
    From $G$, we construct a DES $A=(V\cup\{\clubsuit\},\{a\},\delta,s)$, where $\clubsuit\notin V$ is a new state and $a$ is an observable event. For every $(p,r)\in E$, we add the transition $(p,a,r)$ to $\delta$, and for every $p\in V\setminus\{t\}$, we add the transition $(p,a,\clubsuit)$ to $\delta$. Moreover, we add the self-loop transitions $(\clubsuit,a,\clubsuit)$ and $(t,a,t)$ to $\delta$. The construction is depicted in Fig.~\ref{fig1}. Notice that $A$ is deadlock-free and has no unobservable events. We now show that $t$ is not reachable from $s$ in the graph $G$ if and only if the DES $A$ is strongly (periodically) detectable.

    If node $t$ is not reachable from $s$ in $G$, then, for every $k\ge |V|$, $\delta(s,a^k)=\{\clubsuit\}$. Hence $A$ is strongly (periodically) detectable. 
    
    If $t$ is reachable from $s$, then, for every $k\ge |V|$, $\delta(s,a^k)=\{t,\clubsuit\}$. Hence $A$ is not strongly (periodically) detectable.
  \end{proof}

  We point out that using a unique event for every transition can show NL-hardness for DESs modeled as DFAs.

\section{Conclusions}
  We studied the complexity of deciding detectability of discrete event systems modeled as finite automata and showed that deciding weak (periodic) detectability is intractable for all deadlock-free DESs. On the other hand, we showed that strong (periodic) detectability can be decided efficiently on a parallel computer.


\bibliographystyle{plain}
\bibliography{mybib}

\end{document}